\documentclass[nohyperref]{article}
\usepackage[utf8]{inputenc}

\usepackage[algo2e,ruled]{algorithm2e}
\usepackage{graphicx}
\usepackage{xcolor}
\usepackage{mathtools}
\usepackage{hyperref}
\usepackage{amsmath,amsthm,amsfonts}
\usepackage{bbm}
\usepackage{thmtools} %
\usepackage{xspace}
\usepackage{enumitem}

\usepackage{wrapfig}  

\allowdisplaybreaks

\DeclareMathOperator*{\argmax}{argmax}

\newtheorem{theorem}{Theorem}

\newtheorem{claim}[theorem]{Claim}
\newtheorem{corollary}[theorem]{Corollary}

\newtheorem{lemma}[theorem]{Lemma}

\newcommand{\eb}{b}
\newcommand{\cost}{\mathrm{cost}}

\newcommand{\cA}{\mathcal{A}}
\newcommand{\cD}{\mathcal{D}}

\newcommand{\opt}{\textsc{opt}}
\newcommand{\FiF}{\textsf{FiF}}

\newcommand{\eps}{\epsilon}

\newcommand{\rk}{r}
\newcommand{\RandMarker}{\textsf{RandomMarker}\xspace}
\newcommand{\LRU}{\textsf{LRU}\xspace}
\newcommand{\BlindOracle}{\textsf{BlindOracle}\xspace}
\newcommand{\LVMarker}{\textsf{LVMarker}\xspace}
\newcommand{\RMarker}{\textsf{RohatgiMarker}\xspace}
\newcommand{\RobustOracle}{\textsf{RobustOracle}\xspace}
\newcommand{\AdaptiveQuery}{\textsf{AdaptiveQuery}\xspace}

\newcommand{\E}{\mathbb{E}}

\newcommand{\chain}{C}

\newlength\mylen

\usepackage{multirow}

\usepackage[accepted]{icml2022}

\begin{document}
\onecolumn

\icmltitle{Parsimonious Learning-Augmented Caching} %

\icmlsetsymbol{equal}{*}

\begin{icmlauthorlist}
\icmlauthor{Sungjin Im}{merced}
\icmlauthor{Ravi Kumar}{goog}
\icmlauthor{Aditya Petety}{merced}
\icmlauthor{Manish Purohit}{goog}
\end{icmlauthorlist}

\icmlaffiliation{merced}{University of California, Merced, CA, USA.}
\icmlaffiliation{goog}{Google Research, Mountain View, CA, USA}

\icmlcorrespondingauthor{Sungjin Im}{sim3@ucmerced.edu}
\icmlcorrespondingauthor{Manish Purohit}{mpurohit@google.com}

\icmlkeywords{Parsimonious Queries, Learning-augmented Algorithms, Caching}

\vskip 0.3in

\printAffiliationsAndNotice{\icmlEqualContribution}

\begin{abstract}
Learning-augmented algorithms---in which,  traditional algorithms are augmented with machine-learned predictions---have emerged as a framework to go beyond worst-case analysis. The overarching goal is to design algorithms that perform near-optimally when the predictions are accurate yet retain certain worst-case guarantees irrespective of the accuracy of the predictions. This framework has been successfully applied to online problems such as caching where the predictions can be used to alleviate uncertainties.

In this paper we introduce and study the setting in which the learning-augmented algorithm can utilize the predictions parsimoniously.  We consider the caching problem---which has been extensively studied in the learning-augmented setting---and show that one can achieve quantitatively similar results but only using a \emph{sublinear} number of predictions.
\end{abstract}

\section{Introduction}

Learning-augmented algorithms have recently emerged as a framework to strengthen traditional algorithms with machine learned predictions. Traditional algorithm design obtains algorithms with formal guarantees for \emph{all} inputs.  Hence, they are often geared towards working well on worst-case inputs and not for typical, real-world instances.  In contrast, machine learning performs extremely well on typical instances but can occasionally fail on rare instances.  The learning-augmented framework aims to design algorithms that can benefit from the machine learning predictions while retaining worst-case guarantees. 

This framework was initiated by~\citet{kraska2018case}, who demonstrated that indexed data structures can be improved using learned predictions. Inspired by their work, \citet{lykouris2018competitive} studied the classic online caching problem and obtained an algorithm whose performance guarantee gracefully degrades as the prediction quality worsens but still remains robust regardless of the prediction quality. %
The learning-augmented framework has found applications in streaming algorithms, data structures, and particularly for online algorithms where predictions can alleviate the uncertainties for unseen future inputs; see the survey~\cite{MitzenmacherV20}.

In this paper we focus on an important yet largely overlooked aspect in previous works---the \emph{cost of predictions}. Predictions are typically obtained from an ML model, which can be computationally expensive; this makes it highly desirable to use predictions \emph{parsimoniously}. In this work we study online caching in the learning-augmented framework in which hints are used sparingly.

\paragraph{Online caching.}
In online caching, a sequence of page requests arrive at a cache of size $k$. If the requested page is in the cache, then it can be served at no extra cost, but otherwise a cache miss occurs to fetch the missing page into the cache. The goal of an online algorithm is to minimize the number of cache misses.  A number of randomized algorithms are known to be $\Theta(\log k)$-competitive~\cite{achlioptas2000competitive, fiat1991competitive}, meaning that they incur $O(\log k)$ times more cache misses than the offline optimal solution for all inputs, and this is the best possible. 
Belady's \emph{furthest-in-future} algorithm~\cite{belady} that always evicts the page whose next request is the furthest in the future is well-known to be the optimal offline algorithm

To exploit predictions,~\citet{lykouris2018competitive} proposed an algorithm that assumes the knowledge of the next \emph{predicted} request time of \emph{all} pages in the cache. For the $\ell_1$-norm prediction error with respect to the actual arrival times, they showed that the algorithm's competitive ratio improves to $O(1)$ as the error tends to 0 and remains $O(\log k)$ always. These bounds have further been quantitatively improved recently by~\citet{rohatgi2020near,wei2020better}.

\paragraph{Our contributions.}

We show that we can use significantly fewer predictions for online caching to obtain  results similar to the aforementioned work. More precisely, we allow our algorithm to \emph{query} $b$ pages in cache to learn their predicted next arrival time for each cache miss. 
We show that we can obtain an $O(\log_{b+1} k)$-competitive ratio for good predictions while retaining the  $O(\log k)$-competitive ratio always (Theorem~\ref{thm:main-with-errors}).  Thus, as long as the cache miss rate is $\frac{1}{k^\eps}$ for any constant $\eps > 0$, we can obtain a constant $O(\log \frac{1}{\eps})$-competitive ratio using a \emph{sublinear} number of queries in the number of page requests. 
We also show that our trade-off is near-optimal (Theorem~ \ref{thm:lb}).
Our experiments show that even with very few queries, e.g., making two queries per cache miss, we can significantly improve the traditional online algorithms with the worst case guarantees in practice. The experimental results also demonstrate that we can match (and even exceed) the performance of prior learning-augmented algorithms but querying only $\sim 11\%$ of the page requests.

As is typical for most caching algorithms, our algorithm is also based on the randomized marking algorithm.   However, instead of evicting a randomly chosen page per cache miss, the algorithm queries $b$ unmarked pages in the cache and evicts the one with the furthest predicted request time.  At the high-level, if there are $k$ unmarked pages, we can show that the evicted page is not requested before $k / (b+1)$ pages in the cache  in expectation, provided all the predictions are correct.  Using this we can formally show how to reduce the number of cache misses. While this idea is easy to state, the analysis is delicate as the prediction error is defined only over the pages that were queried. To keep the competitive ratio $O(\log k)$, we follow the technique of \citet{lykouris2018competitive} and switch to using the randomized marking strategy once we detect that the algorithm has made too many mistakes.
The lower bound is shown by an explicit but intricate construction; the formal analysis is quite subtle. 

\paragraph{Related work.}

Online caching has been extensively studied in the literature.  For generalizations of caching, including the $k$-server problem, see \cite{koutsoupias1995k,bansal2015polylogarithmic,bubeck2018k,lee2018fusible}; see also~\cite{bansal2012primal,adamaszek2012log}.  The reader is referred to the book by~\citet{borodin2005online} for a general overview of online algorithms. 

Learning-augmented algorithms largely fall in the rubric of ``beyond worst-case algorithms''; see~\cite{R2020} for an extensive survey of the field. They have recently been extensively explored particularly for online algorithms, including load balancing \cite{LattanziLMV,li2021online}, rent-or-buy \cite{purohit2018improving}, scheduling \cite{AzarLT21}, online set cover \cite{BamasMS20}, metrical task systems \cite{antoniadis2020online}, and many others. For online caching, its weighted version has been studied in \cite{JiangP020,wcaching2}. 

The problem of learning-augmented algorithms with sub-linear number of queries was recently studied by~\citet{BCKP21}, but in the regret setting for online linear optimization.  Our paper studies an analogous question for caching, but in the competitive ratio setting.

\section{Model}

Let ${\cal U}$ denote a universe of pages and $k$ be the number of distinct pages that can be held in the cache at any time. In the classical unweighted \emph{caching} problem, a sequence $\Gamma = \langle p_1, p_2, \ldots \rangle$, where each $p_i \in {\cal U}$, of page requests arrives online and the algorithm is required to maintain a set of at most $k$ pages in the cache at any time. At any time $t$, if the currently requested page $p_t$ is not in the cache, then the algorithm incurs a \emph{cache miss} and must fetch the requested page in the cache (possibly by evicting some other page). The objective of the online algorithm is to minimize the total number of cache misses incurred.

Note that an online algorithm has to choose the page to be evicted without knowing $\Gamma$.  We measure its performance by comparing against Belady's \emph{furthest-in-the-future (\FiF)} algorithm~\cite{belady}, which is the optimal offline algorithm that knows $\Gamma$.  Let $\cost_{\Gamma}(\cdot)$ denote the total number of cache misses of an  algorithm for the request sequence $\Gamma$ and let $\opt = \cost_{\Gamma}(\FiF)$ be the cost of the optimal offline solution.  An online (randomized) algorithm $\cal A$ is said to be \emph{$c$-competitive} if for all request sequences $\Gamma$, it holds that
\[
E[\cost_\Gamma({\cal A})] 
\leq
c \cdot \opt + b,
\]
where $b \geq 0$ is a  constant  independent of the  length  of $\Gamma$, and the expectation is over the randomness (if any) of  ${\cal A}$.  For brevity, from now on we will work with a given $\Gamma$ and omit it from all the subsequent notation.  

In the usual learning-augmented setting, at each time $t$, along with the requested page $p_t$, the algorithm is presented with a (possibly noisy) prediction $\tau_t \in \mathbb{N}$ for the next time after $t$ that the page $p_t$ will be requested again; hence the predicted arrival time of the next request is available for every page in the cache.  In the learning-augmented setting \emph{with queries}, at any time $t$ and for any page $p$ that is in the cache, the algorithm is allowed to \emph{query} a possibly noisy (stochastic) oracle $\cal Q$ for the time, after $t$, of the next request for $p$.  Let $\tau_{p, t} = {\cal Q}(p,t)$ denote such a predicted arrival time of the next request to page $p$ after time $t$;  let $a_{p,t} \geq t$ denote the actual arrival time of the next request to page $p$. Let $Q$ be the set of queries made to ${\cal Q}$. We define the \emph{error} of the oracle to be $\eta = \sum_{(p,t) \in Q} |\tau_{p,t} - a_{p,t}|$. %

The learning-augmented setting with queries generalizes many well-studied caching problems.  On one hand, if the algorithm makes no queries to the oracle, then it is the standard caching problem and we can get a $O(\log k)$-competitive solution, say, with a randomized marking algorithm~(see Section~\ref{sec:marking}).  On the other hand, if the oracle is error-free and the algorithm queries it at every time step, Belady's algorithm yields the optimal solution.  In a recent work,~\citet{lykouris2018competitive,wei2020better,rohatgi2020near} designed a learning-augmented caching algorithm for noisy oracles, showing a tight trade-off between the error of the oracle and the competitive ratio of the algorithm; their algorithm, however, queries the oracle at every time step.  The question we ask in this paper is: can we get similar trade-offs but using much fewer queries?

\section{Preliminaries}
\label{sec:prelim}

A pair $(p,t_1)$ and $(q,t_2)$ of queries in $Q$ is called an \emph{inversion} if $\tau_{p,t_1} \geq \tau_{q,t_2} \text{ but } a_{p,t_1} < a_{q,t_2}$, i.e., the next request of page $p$ is earlier than that of $q$ although the predictions indicated otherwise. Let $I = |\{(p,t_1), (q,t_2) \mid \tau_{p,t_1} \geq \tau_{q,t_2} \text{ but } a_{p,t_1} < a_{q,t_2}\}|$ be the number of inversions. 
The following relates the number of inversions to the  error. 

\begin{lemma}[\citet{Diaconis1977SpearmansFA, rohatgi2020near}]
\label{lem:dg}
For any request sequence $\Gamma$ and any set $Q$ of queries, %
\[\eta \geq \dfrac{1}{2}I.\]
\end{lemma}

\subsection{Marking algorithms}
\label{sec:marking}

Marking algorithms are a class of caching algorithms that associate a ``marking'' bit with each page in the cache, and upon a cache miss only evict an \emph{unmarked} page from the cache. Formally, the algorithm first divides the request sequence into phases where a \emph{phase} is a maximal contiguous sequence of requests to only $k$ distinct pages. At the beginning of each phase, all pages in the cache are unmarked. Pages that are requested during the phase get marked one by one and upon any cache miss, the algorithm only evicts some \emph{unmarked} page. Once all the pages in the cache have been marked, a new phase begins and the process repeats. It is well known that \emph{any} marking algorithm is $O(k)$-competitive and the \emph{randomized marking}~\cite{fiat1991competitive} algorithm, which evicts an unmarked page chosen uniformly at random, is $O(H(k))$-competitive, where $H(k) := 1 + \frac{1}{2} + \cdots + \frac{1}{k} = \Theta(\log k)$.

\begin{algorithm2e}
\For{each requested page $p$} {
\eIf{$p$ in cache}{
``Mark'' $p$
}{
\If{all pages in cache are marked} {
Unmark all pages
}
\emph{Evict} an unmarked page

Fetch $p$ in cache and ``mark'' it
}
}
\caption{A generic \emph{marking} algorithm.}
\label{alg:marking}
\end{algorithm2e}

Consider any phase $h$ and an arbitrary page $p$ that is requested in phase $h$. We say that page $p$ is \emph{clean} if $p$ was not requested in the previous phase (i.e., phase $h-1$), and we say $p$ is \emph{stale} otherwise. Note that once $k$ is known, the phases of the sequence---as well as clean and stale pages---are determined, \emph{independent}  of the algorithm.
Let $\ell_h$ denote the total number of distinct clean pages requested in phase $h$. The following result bounds the number of cache misses incurred by the optimal offline algorithm in terms of the number of distinct clean pages.

\begin{lemma}[\cite{fiat1991competitive}]
\label{lem:marker_clean}
$\frac{1}{2} \sum_{h} \ell_h \leq \opt \leq \sum_h \ell_h$.
\end{lemma}
\section{Warm-up: Modified marking algorithm}
\label{sec:warm-up}

We first show how the classic randomized marking algorithm~\cite{fiat1991competitive} can be modified to effectively use predictions but making fewer queries. For ease of exposition, we assume for now that the oracle is error-free; we extend the analysis to handle noisy predictions in Section~\ref{sec:warmup-handle-error}.

We consider the following modification to the marking algorithm: whenever a page needs to be evicted, if there are at least $\eps k$ \emph{unmarked} pages in the cache, then evict an unmarked page chosen uniformly at random; otherwise, query ${\cal Q}(p, t)$ for all unmarked pages $p$ and evict the page whose next request appears furthest in the future (i.e., apply Belady's method). We remark that once we query all the remaining unmarked pages in a phase, we can simply reuse these predictions for any further cache misses and hence make at most $\epsilon k$ cache misses in any phase.
Algorithm~\ref{alg:naive_eviction} describes this naive eviction policy formally.

\begin{algorithm2e}
\SetAlgoVlined
\SetKwFunction{FRecurs}{Evict}%
\SetKwProg{Fn}{Function:}{\string:}{}
\Fn(){\FRecurs{}}{
\KwData{$U \subseteq {\cal U}$: Set of unmarked pages in cache\\
}
\KwResult{$\alpha$: Page to be evicted}
\BlankLine
\eIf{$|U| \geq \eps k$}
{ $\alpha \leftarrow $ Uniformly random page from $U$}
{
\If{we have not already queried pages in $U$ in this phase}{
\ForEach{page $p$ in $U$}{
 Let $\tau_{p} \leftarrow {\cal Q}(p, t)$
}
}
$\alpha \leftarrow \argmax_{p \in U} \tau_{p}$
}
\KwRet{$\alpha$}\;
}
\caption{Naive eviction.}
\label{alg:naive_eviction}
\end{algorithm2e}

\begin{theorem}
\label{thm:warmup}
For any $\eps > 0$, for any request sequence $\Gamma$, there is an $O(\log(1/\eps))$-competitive algorithm for caching that makes at most $\eps |\Gamma|$ queries.
\end{theorem}

\begin{proof}
Consider any phase $h$ of the marking algorithm and $\ell_h$ be the number of clean pages in that phase. Let $p_1, \ldots, p_k$ denote the $k$ pages in cache at the beginning of the phase and further suppose that the pages are sorted in order of the arrival time of the first request to a page in this phase (breaking ties arbitrarily). In other words, pages $p_1$, \ldots, $p_{k - \ell_h}$ are the $k-\ell_h$ stale pages requested in this phase and further the first request to page $p_i$ is earlier than that of page $p_j$ for any $i < j$.

Consider any stale page $p_i$ where $1 \leq i \leq k - \eps k$, and let $\ell^{(i)} \leq \ell_h$ be the number of clean pages that have been requested before the first request to page $p_i$. When the first request to page $i$ arrives, there are exactly $k - i + 1$ unmarked stale pages of which $\ell^{(i)}$ pages have been evicted from the cache uniformly at random. Hence, the algorithm incurs a cache miss for page $p_i$ with probability $\frac{\ell^{(i)}}{k-i+1} \leq \frac{\ell_h}{k-i+1}$.

Finally, consider the first request to page $p_{k - \eps k}$. If the algorithm incurs any cache miss after this time, then it queries all the $\eps k$ remaining unmarked pages and evicts the page whose next request is furthest in the future, i.e.,  it evicts a page from the set $\{p_{k - \ell_h},\ldots,p_k\}$ that is not requested in this phase. Thus, for any $i > k-\eps k$, page $p_i$ incurs a cache miss only if it has already been evicted by the time page $p_{k - \eps k}$ is first requested. Thus, any such page $p_i$ incurs a cache miss with probability at most $\min \{1, \frac{\ell_h}{\eps k}\}$.

By the linearity of expectation, summing over all pages $p_i$, the expected number of cache misses incurred by the algorithm for stale pages is at most $\sum_{i=1}^{k-\eps k} \frac{\ell_h}{k-i+1} + \sum_{i=k-\eps k + 1}^{k-\ell_h} \frac{\ell_h}{\eps k} \leq \ell_h + \ell_h(H_k - H_{\eps k}) \leq O(\ell_h (\log\frac{1}{\eps}))$. In addition, the algorithm also incurs $\ell_h$ additional cache misses for the clean pages. Hence, the total expected number of cache misses incurred in phase $h$ is $O(\ell_h \log(\frac{1}{\eps}))$. The desired competitive ratio now follows from Lemma \ref{lem:marker_clean}.

To bound the total number of queries, we observe that the algorithm makes at most $\eps k$ queries in each phase. Since each phase has at least $k$ requests, any request sequence $\Gamma$ has at most $|\Gamma|/k$ phases, and thus the total number of queries is at most $\eps |\Gamma|$.
\end{proof}

\subsection{Handling prediction errors}
    \label{sec:warmup-handle-error}

Let us now consider the case where the oracle can give erroneous predictions. 

Since Algorithm~\ref{alg:naive_eviction} does not utilize predictions as long as there are at least $\epsilon k$ unmarked pages left in the cache, we only need to reconsider the cache misses that occur after there are fewer than $\epsilon k$ unmarked pages left. Consider any page $p_i$ for $i > k - \epsilon k$ and let $\tilde t$ denote the time when the algorithm queries all the remaining unmarked pages. Suppose the algorithm incurs a cache miss on page $p_i$ and evicts page $q = \argmax_{p \in U} \tau_{p,\tilde t}$. Now, if the predictions are correct, then $q$ belongs to the set $\{p_{k-\ell_h}, \ldots, p_k\}$ of pages that are not requested in this phase. However, suppose the predictions are incorrect and page $q$ is requested in this phase, then the algorithm incurs an additional cache miss. However, in this case the pair $(q, \tilde t)$ and $(p_k, \tilde t)$ of queries is an inversion and we can charge the additional cache miss incurred to this inversion. Since we only incur at most one cache miss for a page, it can be easily verified that we charge at most one cache miss to a specific inversion. Let $I_h$ be the total number of inversions for queries made in phase $h$, then from the above discussion we have that the expected number of cache misses incurred by the algorithm in phase $h$ is at most $O(\ell_h (\log \frac{1}{\epsilon}) + \E[I_h])$. Hence, the total cost incurred over all phases is $O( (\sum_h \ell_h)(\log \frac{1}{\epsilon}) + \E[I])$ %
where the expectation is over the randomness in the pages evicted by the algorithm. Using Lemma~\ref{lem:dg} and Lemma~\ref{lem:marker_clean}, we conclude that the total cost incurred by Algorithm~\ref{alg:naive_eviction} is at most $O(2\log(1/\epsilon) \opt + \E[\eta])$ and obtain the following:

\begin{theorem}
\label{thm:warmup-with-errors}
For any $\epsilon > 0$, there is an $O(\log(1/\epsilon) + \E[\eta]/\opt)$-competitive algorithm for caching that makes at most $\epsilon |\Gamma|$ queries.
\end{theorem}

While this warm-up result is a proof of concept for parsimonious use of predictions, to achieve a constant competitive ratio, we still need to make a linear number of queries in the request sequence length. To overcome this weakness we propose a new algorithm in the following section that is more adaptive in deciding which pages to query.

\section{Adaptive query algorithm}

The new algorithm queries for $\eb$ unmarked pages uniformly at random per cache miss and evicts the one that is predicted to be requested the furthest in the future. We call this the \emph{adaptive query} algorithm (\AdaptiveQuery-$b$); see  Algorithm~\ref{alg:adaptive_eviction}. Here,  $\eb$ is a parameter that governs a trade-off between the desired competitive ratio and the number of queries we are willing to make per cache miss. 
\begin{algorithm2e}
\SetAlgoVlined
\SetKwFunction{FRecurs}{Evict}%
\SetKwProg{Fn}{Function:}{\string:}{}
\Fn(){\FRecurs{}}
{
\KwData{$U \subseteq {\cal U}$: Set of unmarked pages in cache and an integer $b >0$}
\KwResult{$\alpha$: Page to be evicted}
\BlankLine
$S \leftarrow$ Sample $\eb$ pages from $U$ uniformly at random without replacement \\
Let $\tau_p \leftarrow {\cal Q}(p, t)$ for all $p \in S$\\
$\alpha \leftarrow \argmax_{p \in S} \tau_p $\\
\KwRet{$\alpha$}
}
\caption{Adaptive query eviction.}
\label{alg:adaptive_eviction}
\end{algorithm2e}

As before we first analyze the algorithm assuming the predictions are all correct. We will show the following trade-off in Section~\ref{sec:adaptive-analysis}. 
\begin{theorem}
    \label{thm:caching-adaptive}
    Under the assumption that the oracle is error-free, for any integer $\eb > 0$, the adaptive query algorithm is $2( \log_{\eb+1} k +3)$-competitive and makes at most $2\eb (\log_{\eb+1} k +3) \cdot \opt$ queries in expectation.
\end{theorem}

This bound is shown to be nearly tight in 
 Section~\ref{sec:adaptive-lb}; see Theorem~\ref{thm:lb}. We then extend  Theorem~\ref{thm:caching-adaptive} in Section~\ref{sec:main-handle-error}, so it can handle error-prone predictions.

\subsection{Analysis}
    \label{sec:adaptive-analysis}
 
If we show that the adaptive query algorithm is $c$-competitive, then it immediately follows that the number of queries made is $cb \cdot \opt$. Thus, we only need to establish the desired competitive ratio. 
 
Consider any fixed phase $h$ of the marking algorithm and let $f_1, \ldots, f_{\ell_h}$ be the clean pages requested in that phase. We consider the following notion of eviction chains~\citep{lykouris2018competitive} for the sake of analysis. An \emph{eviction chain} $\chain_i = \langle q_{i,0} := f_i, q_{i,1}, \ldots, q_{i,M_i}\rangle$ is a sequence of pages constructed as follows: $q_{i,1}$ is the stale page that is evicted by the algorithm when it serves the clean page $f_i$; similarly for all $j \geq 1$, $q_{i,j+1}$ is the stale page that gets evicted when the algorithm serves the request to page $q_{i,j}$. Eventually, a stale $q_{i,M_i}$ gets evicted that is not requested in the phase and the sequence ends. We note that each eviction chain starts with a distinct clean page and ends with a stale page that is not requested in the phase. Further, the $\ell_h$ eviction chains are disjoint and each cache miss incurred by the algorithm is encoded in these chains. The $i$th eviction chain $\chain_i$ leads to $M_i$ cache misses where $M_i$ is a random variable. Our goal is to bound the total number of cache misses, i.e., $\E[\sum_{i=1}^{\ell_h} M_i]$.

\paragraph{Page ranks.}
We first order all clean pages and stale pages in the cache by the arrival time of the first request to that page in this phase (the $\ell_h$ stale pages that are not requested in the phase appear last in the ordering, in an arbitrary order).
For each stale page $p$ evicted by the algorithm, we define its \emph{rank} $\rk(p)$ as the number of stale pages after page $p$ in the above ordering that have not yet been evicted (at the time $p$ was evicted). By construction of the eviction chains, page $q_{i,j}$ is evicted when page $q_{i,j-1}$ is requested and hence $q_{i,j}$ is after $q_{i,j-1}$ in the ordering (since all pages before $q_{i,j-1}$ in the ordering have already been marked). Hence, we always have $\rk(q_{i,j}) \leq \rk(q_{i,j-1})$. Similarly, for each clean page $f_i$, we define its rank $\rk(f_i) = \rk(q_{i,0})$ to be the number of stale pages after $f_i$ that have not yet been evicted when $f_i$ was requested. Note that $\rk(f_i) \leq k$, for all $1\leq i \leq \ell_h$.

We first show the following simple lemma that follows from order statistics of the uniform distribution. We defer its proof to the Supplementary Material.
\begin{lemma}
    \label{lem:uniform-ub}
    If $S = \{s_1, \ldots, s_b\}$ is a set sampled uniformly at random without replacement from $\{0, 1, \ldots, \rk\}$, then $\E [\min_{t \in [b]} s_t] \leq \frac{\rk}{b+1}$. 
\end{lemma}

The following two lemmas are used to bound the expected length of an eviction chain.

\begin{lemma}
    \label{lem:page-rank}
    Consider any eviction chain $\chain_i$ and suppose it evicts page $q_{i,j+1}$ to service a request to page $q_{i,j}$. Then we have
    $\E[ \rk(q_{i,j+1}) \mid \rk(q_{i,j})] \leq \frac{\rk(q_{i,j})}{\eb +1}$.
\end{lemma}
\begin{proof}
    When a cache miss occurs for page $q_{i,j}$, note that all pages that appear before $q_{i,j}$ (when ordered by the arrival time of their first request in the phase) have already been marked.
    Thus, all the queried stale pages must appear after $q_{i,j}$. Suppose there are $\rk \leq \rk(q_{i,j})$ unmarked stale pages left. When the predictions are all correct, Algorithm~\ref{alg:adaptive_eviction} randomly samples $b$ pages from all unmarked stale pages and evicts $q_{i,j+1}$ as the one that is latest in the ordering. In other words, $\rk(q_{i,j+1})$ is the minimum of $b$ uniform samples from $\{0, 1, \ldots, \rk-1\}$. Thus from Lemma \ref{lem:uniform-ub}, we have $\E[\rk(q_{i,j+1}) \mid \rk(q_{i,j})] \leq \frac{\rk-1}{b+1} \leq \frac{\rk(q_{i,j})}{b+1}$.
\end{proof}

\begin{lemma}
    \label{lem:length-ub}
    For every $1 \leq i \leq \ell_h$, we have $\E[M_i] \leq \log_{b+1} k + 3$ where $M_i$ is the length of the eviction chain beginning with the clean page $f_i$.
\end{lemma}
\begin{proof}
 Note that an eviction chain ends when it evicts one of the $\ell_h$ stale pages that are not requested in the phase. Fix a particular chain $\chain_i$ and for brevity, let $\rk_j := \rk(q_{i,j})$. Since we have $r_0 \leq k$, using Lemma~\ref{lem:page-rank} and the law of iterated expectation, we have
 $\E[\rk_j] \leq \frac{k}{(b+1)^j}$. By Markov's inequality, we have $\Pr[\rk_j \geq 1] \leq \frac{k}{(b+1)^j}$. Note that if $M_i > j$, then it must be the case that $\rk_j \geq 1$. This is because if $\rk_j = 0$, then $q_{i, j}$ will not be evicted and the chain $C_i$ must have length $j$. 
 
 Let $c := \lceil \log_{b+1} (k) \rceil$. We can now bound the expected length of the chain as follows.
 \begin{align*}
     \E[M_i] &= \sum_{j = 0}^{c-1} \Pr[M_i \geq j] 
      + \sum_{j \geq c} \Pr[M_i \geq j] \\
     &= c + \sum_{j \geq 0} \Pr[M_i > c+j] 
     = c +  \sum_{j \geq 0} \Pr[ \rk_{c+j} \geq 1] \\
     &\leq c + \sum_{j \geq 0} \frac{k}{(b+1)^{c+j}} = c + \sum_{j \geq 0} \frac{1}{(b+1)^{j}}  \\
     &\leq \log_{b+1} (k) + 3. \qedhere
 \end{align*}
\end{proof}

\begin{proof}[Proof of Theorem \ref{thm:caching-adaptive}]
Fix a phase $h$ of the marking algorithm. Since every cache miss incurred by the algorithm is recorded in exactly one eviction chain, the expected total number of cache misses incurred by the algorithm in this phase is $\E[\sum_{i=1}^{\ell_h} M_i]$. Using Lemma \ref{lem:length-ub}, this is at most $\ell_h(\log_{b+1} k + 3)$. The desired competitive ratio now follows from Lemma \ref{lem:marker_clean}. 
Further, the algorithm makes at most $b$ queries for each cache miss it incurs and thus we have Theorem \ref{thm:caching-adaptive}.
\end{proof}

\subsection{Handling prediction errors}
\label{sec:main-handle-error}

In this section we extend the analysis to allow for oracles that make erroneous predictions. In this scenario, since we evict a page that is only \emph{predicted} to arrive furthest in the future (and not actually be the one to arrive the latest), Lemma~\ref{lem:page-rank} fails. However, as we show below, in this case the oracle has a large error and we can bound the expected cost of the algorithm in terms of the prediction error.

We first show the following technical statement that relates the rank of the page evicted by the algorithm and the rank of the page that actually arrives the furthest in the future from among the sampled pages (while processing any cache miss).

\begin{lemma}
\label{lem:page-rank-error}
Let $S$ be any set of $b$ pages and let $a_1 < \dots < a_b$ denote their actual next arrival times and let $\langle \tau_1, \dots, \tau_b\rangle$ be the sequence of their predicted arrival times. Let $\eta_S = \sum_{i=1}^b |a_i - \tau_i|$ be $\ell_1$-error of the predictions for the set $S$. If $\hat{b} = \argmax_\alpha \tau_\alpha$ is the page with the furthest predicted arrival time, then we have
\[ \rk(\hat b) \leq \rk(b) + \eta_S.\]
\end{lemma}

\begin{proof}
We assume that $\hat b \neq b$ since otherwise the lemma is trivial. By definition of $\hat b$ we have $\tau_{\hat b} \geq \tau_b$ and $a_{\hat b} < a_b$. For convenience let $\hat p$ and $p$ denote the corresponding pages. 
Since the number of unmarked pages between $\hat p$ and $p$, when ordered by the request time of their first request, is at most $a_{b} - a_{\hat b}$, by definition of rank we have \[\rk(\hat b) \leq \rk(b) + a_b - a_{\hat b}.\]
We now show $\eta_S \geq a_b - a_{\hat b}$, which will complete the proof. To show this, we observe that $\eta_S \geq |a_{\hat b} - \tau_{\hat b}| + |a_b - \tau_b|$. We consider three cases.
    
    \emph{Case 1: $\tau_{\hat b} \leq a_{\hat b}$.} %
    In this case we have $\eta_S \geq |a_b - \tau_b| = (a_b - a_{\hat b}) + (a_{\hat b} - \tau_{b}) \geq a_b - a_{\hat b}$.
    
    \emph{Case 2: $a_b \geq \tau_{\hat b} > a_{\hat b}$.} Since $\tau_b \leq \tau_{\hat b}$, we have 
$\eta_S \geq |a_{\hat b} - \tau_{\hat b}| + |a_b - \tau_b| =  \tau_{\hat b} - a_{\hat b} + a_b - \tau_b \geq a_b - a_{\hat b}$.

    \emph{Case 3: $\tau_{\hat b} > a_b$.} Here we have $|a_{\hat b} - \tau_{\hat b}| = \tau_{\hat b} - a_{\hat b} \geq a_b - a_{\hat b}$.
\end{proof}

Lemma \ref{lem:page-rank-error} lets us  prove the following analog of Lemma \ref{lem:length-ub}.
\begin{lemma}
    \label{lem:length-ub-error}
    For every $1 \leq i \leq \ell_h$, we have $\E[M_i] \leq \log_{b+1} k + 3 + 2\E[\eta_{S_i}]$ where $M_i$ is the length of the eviction chain beginning with the clean page $f_i$ and $S_i$ is the set of pages queried when pages on path $P_i$ are evicted. 
\end{lemma}
\begin{proof}
  Fix a particular chain $\chain_i$. Let $\rk(q_{i,j})$ be the number of stale unmarked pages left in the cache when the algorithm incurs a cache miss for page $q_{i,j}$ at some time $t$. In this case, we sample a set $S$ of $b$ of those pages uniformly at random, and set $q_{i,j+1} = \argmax_{p \in S} {\cal Q}(p, t)$. Let $q^*_{i,j+1} = \argmax_{p \in S} a_{p,t}$ be the sampled page that actually arrives furthest in the future. Then by Lemma \ref{lem:page-rank}, we have $\E[\rk(q^*_{i, j+1}) \mid \rk(q_{i, j})] \leq \frac{\rk(q_{i,j})}{b+1}$. Further, for any queried set $S$ of pages, by Lemma~\ref{lem:page-rank-error} we have, $\rk(q_{i,j+1}) \leq \rk(q^*_{i,j+1}) + \eta_S$ where $\eta_S$ is the $\ell_1$-error of the predictions for the set $S$. Thus we obtain the following where $\eta_{i,j+1}$ is defined to be the prediction error of the oracle for pages queried while evicting page $q_{i,j+1}$.
  \begin{align*}
      \E[\rk(q_{i,j+1}) \mid \rk(q_{i,j})] &\leq \frac{\rk(q_{i,j})}{b+1} + \E[\eta_{i,j+1}].
      \intertext{Now, since we have $r(q_{i,0}) \leq k$, using the law of iterated expectation we have}
      \E[\rk(q_{i,j})] \leq \frac{k}{(b+1)^j} &+ \sum_{j' = 1}^{j}\frac{\E[\eta_{i,j'}]}{(b+1)^{j-j'}}.
      \intertext{Finally, using Markov's inequality, we have}
      \Pr[M_i > j] \leq \Pr[\rk(q_{i,j}) \geq 1] &\leq \frac{k}{(b+1)^j} + \sum_{j' = 1}^{j}\frac{\E[\eta_{i,j'}]}{(b+1)^{j-j'}}.
    \end{align*}
      As earlier, let $c = \lceil \log_{b+1}(k)\rceil$. We can now bound the expected length of the chain.
      \begin{align*}
      \E[M_i] &= c + \sum_{j \geq 0} \Pr[M_i > c+j]\\
      &\leq c + \sum_{j \geq 0} \left( \frac{k}{(b+1)^{c+j}} + \sum_{j' = 1}^{c+j}\frac{\E[\eta_{i,j'}]}{(b+1)^{c+j-j'}} \right)\\
      &\leq c + \sum_{j \geq 0} \frac{1}{(b+1)^j} + \sum_{j' \geq 1} \E[\eta_{i,j'}] \sum_{j \geq j'} \frac{1}{(b+1)^{j - j'}}\\
      &\leq \log_{b+1}(k) + 3 + 2\sum_{j' \geq 1} \E[\eta_{i,j'}].
      \end{align*}
      Finally, since the algorithm makes a distinct set of queries when evicting any page, we have %
      $\sum_{j'\geq 1} \E[\eta_{i,j'}] = \E[\eta_{S_i}]$ and the lemma follows.
\end{proof}

\subsection{Adding worst-case guarantees}
In this section we show how a simple modification to the algorithm allows us to obtain an $O(\log k)$-competitive ratio even when the prediction error is arbitrarily large. 
In order to obtain this worst-case guarantee, we make the following modification: when processing a cache miss for the $j$th page ($q_{i,j}$) on chain $\chain_i$, if $j > \log k$, then the algorithm switches to evict an unmarked stale page uniformly at random (as opposed to querying $b$ pages and evicting the one with the furthest predicted arrival). 

\begin{theorem}
\label{thm:main-with-errors}
For any integer $b > 0$, there is an $O(\min\{\log_{\eb+1} k + \E[\eta]/\opt, \log k\})$-competitive algorithm for caching that makes at most $b$ queries per cache miss.
\end{theorem}

\begin{proof}
When pages are evicted according to Algorithm \ref{alg:adaptive_eviction}, Lemma \ref{lem:length-ub-error} shows that the expected length of any eviction chain is at most $O(\log_{b+1}k + \E[\eta_{S_i}])$. We consider the modified algorithm that switches to evicting a uniformly random unmarked page once the chain length exceeds $\log k$. Following the traditional analysis of the randomized marking algorithm~\cite{fiat1991competitive,lykouris2018competitive}, we observe that once the algorithm switches to random evictions, the length of the chain increases by at most $O(\log k)$ in expectation. Consequently, the modified algorithm incurs at most $O(1) \cdot \min\{\log_{b+1}k + 3 + 2\E[\eta_{S_i}], 2\log k\}$ cache misses in expectation on each eviction chain.

Summing over all clean pages seen in the phase, the expected number of cache misses incurred in any phase $h$ is at most $O(\min\{\log_{b+1}k + \E[\eta_h],\log k\})$, where $\eta_h$ is defined to be the $\ell_1$-error of all the queries made in this phase. The desired competitive ratio now follows from Lemma \ref{lem:marker_clean}.
\end{proof}

\section{Lower bound}
    \label{sec:adaptive-lb}

\newcommand{\cP}{\mathcal{P}}
\newcommand{\cQ}{\mathcal{Q}}
\newcommand{\dd}{\texttt{d}}

The lower bound instance is fairly simple. Each phase starts with a request for a clean page that has never been requested before. Then, it is followed by requests for $k-1$ stale pages that are chosen uniformly at random among the $k$ stale pages.  We provide a formal description below.
 
\paragraph{Lower bound instance.} The page requests proceed in phases. 
Let $P_h$ denote the pages that are requested in phase $h$ for $h \in [H]$, where $H$ is a sufficiently large integer. We will have $|P_h| = k$ for all $h \in [H]$, and $|P_{h+1} \setminus P_h|  = |P_{h} \setminus P_{h+1}| = 1 $ for all $h \in [H-1]$. First, $P_1$ is an arbitrary set of $k$ pages and there is one request for each page in $P_1$. We now iteratively construct $P_{h+1}$ from $P_h$ as follows: Let $f_{h+1}$ be a clean page that has never been requested before. Let $p_1, \ldots, p_k$ be a uniformly random permutation of the set of pages in $P_h$. Then the request sequence for phase $h+1$ is $f_{h+1}$, $\langle f_{h+1}  \rangle^k p_k$, $ \langle f_{h+1} p_{k} \rangle^k p_{k-1}$, \ldots, 
$\langle f_{h+1} p_{k} \ldots p_3 \rangle^k p_{2}$ in this order. Here, $\langle S \rangle^k$ implies $k$ repetitions of the sequence $S$. Focusing on the page arriving after the repeated sequence, we will say that pages are requested in the order of $f_{h+1}, p_{k}, \ldots, p_2$.

The proof of Theorem \ref{thm:lb} requires care to impose constraints on the structure of candidate algorithms, and formally demonstrate that a learning-augmented algorithm for caching can do no better than querying unmarked stale pages and always evict the one that arrives furthest in the future. Unlike in the analysis of the upper bound, the algorithm can make varying numbers of queries per cache miss, even stochastically, which renders the analysis considerably more challenging. We defer the full proof to the Supplementary Material.

\begin{theorem}
    \label{thm:lb}
    For any integer $c \leq \ln k$, any $(c+4)$-competitive algorithm must make at least $\frac{1}{12 \ln (k+1)}c k^{1/ c} \cdot \opt$ queries (with no error). 
\end{theorem}

\section{Experiments}
    \label{sec:exp}

We experimentally evaluate our algorithm on a real-world dataset and demonstrate the empirical dependendence of the competitive ratio on the number of queries made as well as on the prediction errors.

\paragraph{Input dataset.} We use the CitiBike dataset, closely following  \cite{lykouris2018competitive}. The dataset comes from a publicly-available~\cite{citibike} bike sharing platform operating in New York City. For each month of year the 2018, we construct one instance where each page request corresponds to the starting point of a bike trip.  We truncate each months data to the first 25,000 events, and thus each input sequence length is 25,000. Finally we set the cache size $k = 500$ and obtain 7 non-trivial instances\footnote{The other 5 sequences have less than 500 distinct pages and the caching problem is trival.}.
We use a bigger cache than \cite{lykouris2018competitive} to illustrate our algorithm's trade-off between number of queries and the competitive ratio. 

\paragraph{Predictions.}
To demonstrate the empirical dependence of different algorithms on the prediction error, we generate the following synthetic predictions. For each page $p$ in the cache, its predicted next request time is set to its actual next request time plus a noise, which is drawn i.i.d. from a lognormal distribution whose underlying normal distribution has mean 0 and standard deviation $\sigma$. If the page is never requested in the future, we pretend its actual request time is the sequence length plus 1, i.e., 25,001.

We also use a very simple prediction model to demonstrate the efficacy of easy off-the-shelf predictors. For each page, we compute the average time $\mu_p$ elapsed between consecutive requests for that page. For any page $p$ at time $t$, we set the predicted arrival time as $Q(p,t) = \tilde t_p + \mu_p$ where $\tilde t_p$ is the last time before $t$ when page $p$ was requested. We refer to these predictions as ``Mean Predictions'' in Table \ref{tab:results}.

\begin{figure}[tb]
    \centering
    \includegraphics[width=0.45\textwidth]{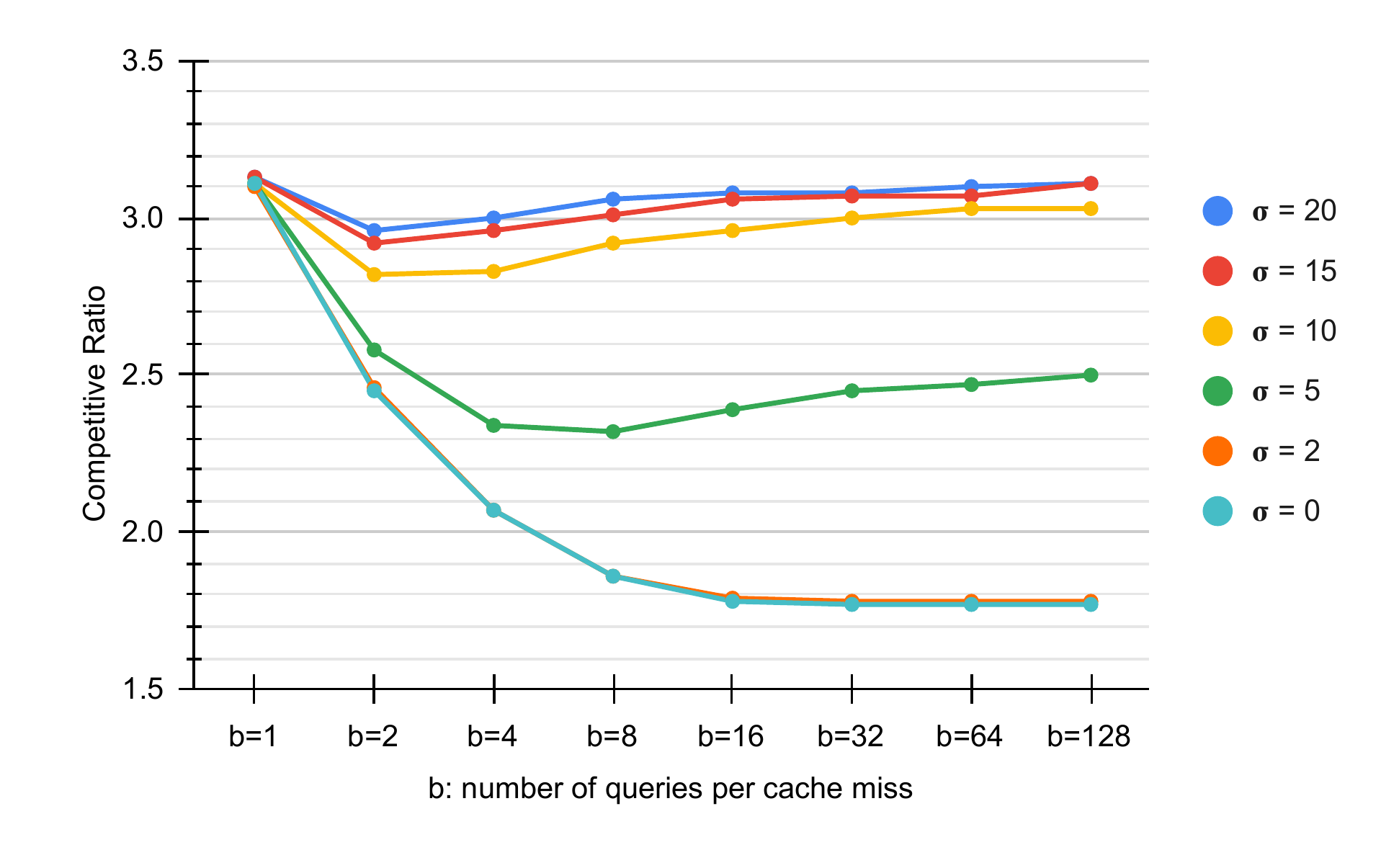}
    \caption{Average competitive ratio of \AdaptiveQuery  for different values of $b$ and the error parameter $\sigma$.}
    \label{fig:exp}
\end{figure}

\paragraph{Algorithms.}  We implement the following algorithms. 

\begin{itemize}[nosep]
    \item \RandMarker (Randomized Marking, \citet{fiat1991competitive}). Evicts a randomly chosen unmarked page; $\Theta(\log k)$-competitive. 
    \item \LRU (Least Recently Used). A widely used heuristic that evicts the least recently used page. 
    \item \BlindOracle. Evicts the page with the latest predicted next request time. 
    \item \LVMarker \cite{lykouris2018competitive}. A learning-augmented marking algorithm that evicts the page with the furthest predicted arrival until the length of the eviction chain is $O(\log k)$ and then switches to evicting a randomly chosen unmarked page.
    \item \RMarker \cite{rohatgi2020near}. Identical to \LVMarker except the switch occurs after the chain length exceeds one. 
    \item \RobustOracle \cite{wei2020better}. Uses the \emph{combiner}~\cite{fiat1994competitive} to combine \BlindOracle and \RandMarker.
    \item \AdaptiveQuery-$b$. Our algorithm (with worst-case guarantees) that is parameterized by $b$, the number of queries made per cache miss. 
\end{itemize}

\paragraph{Results.}

\begin{table}[tb]
\centering
\scriptsize
\caption{Average competitive ratio of algorithms for different error parameters. (Smaller values means better performance.)}
\label{tab:results}
\begin{tabular}{r|c|c|c|c|c}
\hline
\multirow{2}{*}{Algorithms} & Mean  & \multicolumn{4}{c}{Synthetic Predictions}\\ \cline{3-6}
    & Predictions & $\sigma = 0$ & $\sigma = 2$ & $\sigma = 4$ & $\sigma = 6$ \\
\hline    
\RandMarker & 3.14 & 3.14 & 3.14 & 3.14 & 3.14 \rule{0pt}{2.6ex}\\
\LRU   & 2.86 & 2.86  & 2.86  & 2.86  & 2.86  \\
\hline
\BlindOracle & 1.92 & 1.00 & 1.02 & 3.92 & 4.15 \rule{0pt}{2.6ex}\\
\LVMarker & 2.49 & 1.77 & 1.81 & 2.94 & 3.11 \\
\RMarker & 2.54 & 1.77 & 1.83 & 3.15 & 3.29 \\
\RobustOracle & 4.29 & 1.80 & 1.83 & 4.48 & 4.51 \\
\hline
\AdaptiveQuery-2 & 2.91 & 2.46 & 2.46 & 2.52 & 2.65 \rule{0pt}{2.6ex}\\
\AdaptiveQuery-4 & 2.71 & 2.07 & 2.07 & 2.20 & 2.49 \\
\AdaptiveQuery-8 & 2.59 & 1.86 & 1.86 & 2.07 & 2.54 \\
\hline
\end{tabular}
\end{table}

Table~\ref{tab:results} shows the competitive ratios of all  the implemented algorithms averaged over the seven instances.
We observe that our \AdaptiveQuery algorithm performs significantly better that \RandMarker and \LRU (that do not use any predictions), even while using very few predictions, e.g., making $b = 2$ queries on each cache miss. Since our algorithm is equivalent to \RandMarker when $b=1$, this demonstrates than even minimal predictions can considerably help online algorithms.

At the same time, Table \ref{tab:results} also demonstrates that \AdaptiveQuery performs as well as the three other learning-augmented algorithms even for relatively small values of $b$ with both synthetic predictions as well as the simple mean predictions. In fact, for high prediction errors, the \AdaptiveQuery algorithm is less affected by these errors and outperforms the other learning-augmented algorithms. 
For comparison, with $b = 8$, the \AdaptiveQuery algorithm uses only $2839$ queries for each instance on average, it utilizes predictions for about 11\% of requests in the sequence. %

We also compare the dependence of our algorithm on the number of queries and the prediction error. Figure \ref{fig:exp} shows the competitive ratio of \AdaptiveQuery-$b$ for different values of $b$ and different error parameters. Unsurprisingly, we observe that when predictions are perfect (or very good), the competitive ratio of the algorithm improves with the number of queries it is allowed to make. Surprisingly, however, when the prediction error is large, using more queries actually leads to a worse competitive ratio. This is because when predictions are highly erroneous, using many queries often leads to making poor eviction decisions.

\section{Conclusions}

In this paper we initiate the study of online algorithms augmented with parsimonious learned predictions. Both the theory and experimental results suggest that performance of online algorithms can be significantly improved by judiciously using just a few predictions. Such an approach can make learning-augmented algorithms more practically appealing since obtaining predictions is often computationally expensive.
An interesting future direction is to further explore this parsimonious  model for other online problems. For example, consider problems that involve predictions of locations, such as metric task system and online matching \cite{antoniadis2020online}.  It is conceivable that one can use less predictions by spatial interpolation.

\bibliographystyle{icml2022}
\bibliography{bibfile}

\newpage
\onecolumn
\section*{\centering Supplementary Material}
\appendix
\section{Proof of Lemma \ref{lem:uniform-ub}}
\begin{proof}
    We proceed assuming that replacement is allowed since it only increases 
    $\E[\min_{t \in [b]} s_t]$. Then, the sampling process can be simulated by sampling $s'_1, \ldots, s'_b \sim \text{Unif}(0, \rk)$ from the uniform distribution and taking the floor of them. Thus, we have
    $\E[\min_{t \in [b]} s_t] \leq \E [\min_{t \in [b]} s'_t]$.
    
    It well known that if $x_1, \ldots, x_b$ are uniformly sampled from $[0, 1]$, then $\E [\min_{t \in [b]} x_t] = \frac{1}{b+1}$.  (This can be easily verified by observing that $\Pr [\min_{t \in [b]} x_t > x] = (1 - x)^b$ and a simple calculus.)  Since we can set $s'_t = \rk\cdot x_t$, we have,
    $\E [\min_{t \in [b]} s_t] \leq \E [\min_{t \in [b]} s'_t] = \frac{\rk}{b+1}$.
\end{proof}

\section{Lower bound}

We repeat the lower bound instance here for clarity.

\paragraph{Lower bound instance.} The page requests proceed in phases. 
Let $P_h$ denote the pages that are requested in phase $h$ for $h \in [H]$, where $H$ is a sufficiently large integer. We will have $|P_h| = k$ for all $h \in [H]$, and $|P_{h+1} \setminus P_h|  = |P_{h} \setminus P_{h+1}| = 1 $ for all $h \in [H-1]$. First, $P_1$ is an arbitrary set of $k$ pages and there is one request for each page in $P_1$. We now iteratively construct $P_{h+1}$ from $P_h$ as follows: Let $f_{h+1}$ be a clean page that has never been requested before. Let $p_1, \ldots, p_k$ be a uniformly random permutation of the set of pages in $P_h$. Then the request sequence for phase $h+1$ is $f_{h+1}$, $\langle f_{h+1}  \rangle^k p_k$, $ \langle f_{h+1} p_{k} \rangle^k p_{k-1}$, \ldots, 
$\langle f_{h+1} p_{k} \ldots p_3 \rangle^k p_{2}$ in this order. Here, $\langle S \rangle^k$ implies $k$ repetitions of the sequence $S$. Focusing on the page arriving after the repeated sequence, we will say that pages are requested in the order of $f_{h+1}, p_{k}, \ldots, p_2$.

We first observe that the optimum offline solution for such an instance always incurs a cache miss on the first, clean page of each phase (except the first phase) and evicts the unique page in $P_{h-1} \setminus P_h$. By construction, the optimum algorithm incurs no more cache misses in each phase. Finally, any algorithm must incur $k$ cache misses for the first phase and we have the following claim.

\begin{claim}
    \label{claim:lb-opt}
There is an offline solution that incurs exactly $k+H-1$ cache misses in total.
\end{claim}

We would like to lower bound the number of cache misses incurred by any $c$-competitive algorithm. Consider a fixed optimum online algorithm $\cA$. 

\begin{claim}
    \label{claim:lb-simplifying1}
    At the beginning of each phase $h \geq 2$, we can assume without loss of generality that $\cA$ has all pages in $P_{h-1}$ in cache. 
\end{claim}
\begin{proof}
    Assume that the universe of  pages is infinite. Then, knowing that we cannot guess the clean page that will be requested in phase $h$ and all the other requests are for the stale pages, the claim follows.
\end{proof}

Thanks to Claim~\ref{claim:lb-simplifying1} and the repeated identical structure of the lower bound instance in every phase, we can assume without loss of generality that we use the same optimum online algorithm that we call $\cA$ in all phases except the first. If $c'$ is the expected number of cache misses $\cA$ incurs in each phase $h \geq 2$, for $\cA$ to be $c$-competitive, it must be the case that $\frac{k + c'(H-1)}{k + H-1} \leq c$ from Claim~\ref{claim:lb-opt}. Here, $k$ in the numerator is the number of cache misses incurred by $\cA$ in the first phase. Thus, we must have $c' = c$ as $H \rightarrow \infty$.

Therefore, we can focus on one phase and lower bound the expected number of queries made by $\cA$ assuming that it incurs at most $c$ cache misses in expectation. Henceforth we drop indices referring to phases from the notation. We will say that a page is marked in the phase if it has been requested in the phase. For simplicity, we will assume that $\cA$ makes at least one query before each page eviction; this would have no effect on the asymptotic lower bound we aim to prove.

\begin{lemma}
    \label{lem:lb-simplifying2}
    In a phase that is not the first, with a given limit $c > 0$ on the expected number of cache misses, there is an algorithm that makes the minimum number of queries in expectation and simultaneously satisfies the following: 
    \begin{enumerate}[nosep]
        \item[(i)] evicts a page only when it is forced to do so;
        \item[(ii)] never evicts marked pages and therefore it only needs to query unmarked pages; 
        \item[(iii)] only queries pages just before a page eviction;
     \item[(iv)] evicts the queried page with furthest arrive time, if it makes any queries.
    \end{enumerate}
\end{lemma}

We first prove (i)--(iii). After setting up additional notation that will be used throughout the analysis, we will prove (iv).

\begin{proof}[Proof of Lemma~\ref{lem:lb-simplifying2}(i)--(iii)]
    As argued in \citet{fiat1991competitive}, we can assume wlog that the algorithm needs to make a change (evict a page) only when it is forced; such algorithms are called \emph{lazy}. This implies (i).
    
    We now show (ii). Suppose algorithm $\cA$ evicts a page $p_i$ for some $i > j$ to service $p_j$ in the request sub-sequence, $\langle p_{k+1}p_k \ldots p_{j+1}\rangle^k p_j$.
    Then, for every repetition of $\langle p_{k+1}p_k \ldots p_{j}\rangle$ in the subsequence of $\langle p_{k+1}p_k \ldots p_j\rangle^k p_{j-1}$, until $\cA$ fetches $p_i$ and evicts an unmarked page, it incurs another cache miss. If it makes a cache miss for every repetition, we can make $\cA$ better or no worse by instead evicting an arbitrary unmarked page (such an algorithm incurs at most $k$ cache misses even without using randomization). Otherwise, that is, if $\cA$ ever replaces with an unmarked page before $p_{j-1}$ is requested, we could let $\cA$ have done so earlier to reduce cache misses. Thus, we can assume that before we see $p_{j-1}$, an unmarked page was evicted to service $p_j$. This will incur no cache misses for the repetition $\langle p_{k+1}p_k \ldots p_j\rangle^k$. We have shown that we can assume wlog that $\cA$ never evicts marked pages; thus, we have  (ii). 
    
    Now (iii) follows as once a page gets marked it stays marked in the phase. Thus, by deferring the queries until being forced to evict a page, $\cA$ can only potentially avoid querying about pages that will be marked soon. 
\end{proof}

For the remaining analysis, we take the eviction chain view we used in the analysis of the adaptive query algorithm. As our analysis will require careful conditioning and deconditioning, we slightly override the notation. From the above reasoning, particularly from  Lemma~\ref{lem:lb-simplifying2}~(i)--(iii), we now have the following problem: We will see a request sequence for a clean page $p_{k+1}$, and $k-1$ stale pages, $p_{\pi(k)}$, $p_{\pi(k-1)}$, \ldots, $p_{\pi(2)}$. Note that $p_{\pi(1)}$ is the dead page that is not requested. Here,  $\pi(k), \ldots, \pi(1)$, denotes a random permutation of the $k$ stale pages. Then, we consider the chain $\chain$ that starts with $p_{k+1}$ and ends with $p_{\pi(1)}$. Recall that an edge from $q$ to $q'$ means that we evict page $q'$ to service page $q$. Our goal is to lower bound the number of queries made under the requirement that the expected length (number of cache misses) $M$ of $\chain$ is at most $c$. Page $p_{\pi(j)}$ is defined to have rank $j$; this definition is slightly simpler than the  one in Section~\ref{sec:adaptive-analysis} as we have only one clean page, thus only one chain.

With this notation set up, we are now ready to prove Lemma~\ref{lem:lb-simplifying2}(iv).

\begin{proof}[Proof of Lemma~\ref{lem:lb-simplifying2}(iv)]
 Our goal is to consider any algorithm $\cA$ satisfying (i)--(iii), and to construct another algorithm $\cA'$ satisfying (iv) as well, without increasing the number of cache misses but making no more queries. 
 
 In the execution of $\cA$, suppose $i$th evicted page by $\cA$ has rank $R_i$ and $\cA$ makes $q_i$ queries just before the eviction. Now $\cA'$ makes the same number $q_i$ of queries just before evicting $i$th page, but it instead evicts the one with smallest rank among the queried pages, i.e., satisfies property (iv). By a simple induction on $i$, we can show that $\cA'$ generates a sequence that stochastically dominates what $\cA$ generates. More precisely, suppose $\cA$ has made $q$ queries, and the queried pages have ranks $S_1 < \cdots < S_q$. Then, by definition, $\cA'$ has also made $q$ queries (if there are not enough pages to query, then it is only better for $\cA'$), and let $S'_1 < \cdots < S'_q$ be the rank of pages queried by $\cA'$. Then, we say that the sequence $S'$ stochastically dominates  sequence $S$ if $S'_j \leq S_j$ for all $j \in [q]$. Because the chain $\chain$ ends once the last page of rank $1$ is evicted (we can assume we evict only a queried page under the assumption we query at least one page before each page eviction), $\cA'$ make no more cache misses than $\cA$ in expectation. Further, by construction, $\cA'$ can only make less queries than $\cA$ in expectation. 
\end{proof}

Henceforth, we consider an algorithm $\cA$ that satisfies Lemma~\ref{lem:lb-simplifying2}(i)--(iv). Let $R_0 := k+ 1$, be the rank of the clean page. Let $R_i$ denote the rank of the $i$th evicted page. As mentioned before, $\chain$ eventually ends with $p_{\pi(1)}$. For notational convenience, once $R_i$ becomes 1, we define all the subsequent $R_{i+1}, R_{i+2}$, \ldots, to be 1.

Let $Q_i$ denote the number of queries $\cA$ makes when we witness the $i$th cache miss. For analysis, we will assume that we do not make too many queries for each page eviction. This is because we can find the page $p_{\pi(1)}$ without making many more queries. 

\begin{lemma}
    \label{lem:quadratic}
Suppose we show that any algorithm that incurs at most $c$ cache misses in expectation makes at least $d$ queries under the assumption that $R_i > Q_{i+1}^2$ for all $i$. Then, it implies the following lower bound: any algorithm that incurs at most $c+4$ cache misses in expectation makes at least $d/5$ queries. 
\end{lemma}
\begin{proof}
    If the fixed algorithm considered to show the lower bound makes $Q_{i+1}$ queries such that $R_i \leq Q_{i+1}^2$ for the first time, then instead we let it make $\sqrt {R_i}$ queries per cache miss until the phase ends.
    This is equivalent to a problem where the cache size is $R_i$ and we make $\sqrt R_i$ queries per cache miss. Thus, 
    by Lemma~\ref{lem:length-ub}, we know that the number of cache misses is at most 5 in expectation, and we make at most $5 \sqrt R_i \leq 5Q_{i+1}$ queries in expectation. In summary, this change makes at most 4 extra cache misses in expectation and increases the number of queries by a factor of 5. 
\end{proof}

Let us fix $i$ and let $R_{i-1} = r$ and $Q_i = q$. If we choose $q$ points uniformly at random from $[0, r]$ and the expected minimum of the samples is well known to be $\frac{r}{q + 1}$.  But, we want to know the expected value of 
$\frac{r}{R_{i+1}}$, where the only randomness comes from $\pi(1), \ldots, \pi(r-1)$, which we denote as $\pi([r-1])$ for short. Bounding this quantity needs more care. 

Now, our concern is to upper bound $\E_{\pi([r-1])} \frac{r}{S}$, where $S_1, \ldots, S_q$ are sampled uniformly from $[r-1]$ without replacement and $S := \min_{j \in [q]} S_j$. This corresponds to making queries about $q$ pages among $r - 1$ unmarked ones and evicting the one with the minimum rank, i.e., the furthest request time in the future. For ease of analysis, we pretend that samples $S'_1, \ldots, S'_q$ are made from $[1/r, 1]$ and we want to upper bound $\E [\frac{1}{S'}]$, where $S' := \min_{j \in [q]} S'_j$. 
Here, we relax the random selection by making the sampling domain continuous.  

We first show this relaxation does not change the expectation by much. 

\begin{lemma}
    \label{lem:ratio-bd1}
    If $1 \leq q^2 < r$, we have $\E_{\pi(r-1)} [\frac{r}{S}] \leq %
    6 \E [\frac{1}{S'}]$.
\end{lemma}
\begin{proof}
    Scale down $S_1, \ldots, S_q$ by a factor of $r$, so we can pretend that they are sampled from $D  = \{\frac{1}{r}, \ldots, \frac{r-1}{r}\}$ without replacement. Now, we want to upper bound $\E [\frac{1}{S}]$. Let $\cD := {D \choose q}$. 
    
    We observe that 
    \begin{equation*}
     \E [\frac{1}{S}] = \E [\frac{1}{\bar S'} \; \mid \; \mathcal{\bar S}' \in \cD], 
    \end{equation*}
    where $\mathcal{\bar S}' := \{\lfloor r S'_i \rfloor / r \; | \; i \in [q]\}$ and $\bar S' = \min \mathcal{\bar S}'$.
         In other words, after ``rounding'' down each $S'_i$ to the nearest multiple of $1 / r$, if they are all distinct, we keep them. This is an equivalent way of getting samples $S_1, \ldots, S_q$. 
    
    Further, the rounding changes the expectation by a factor of at most 2. Therefore, we have,
    \begin{equation*}
     \E [\frac{1}{\bar S'} \; \mid \; \mathcal{\bar S}' \in \cD] \leq 2\E [\frac{1}{ S'} \; \mid \; \mathcal{\bar S}' \in \cD].
    \end{equation*}
    
    We would like to decondition on $\bar S' \in \cD$. 
    \begin{equation*}
    \E [\frac{1}{ S'} \; \mid \; \mathcal{\bar S}' \in \cD] \cdot \Pr[\mathcal{\bar S}' \in \cD]
    \leq  \E [\frac{1}{ S'}] 
    \end{equation*}

As $\mathcal{\bar S}'$ is a uniform sample with replacement, we have, 
$\Pr[\mathcal{\bar S}' \in \cD] = \frac{r (r-1) \cdots  (r - q+1)}{r^q} \geq (1 - q/r)^q \geq ( 1- q/ (q^2 +1))^q \geq 1/3$, where the last inequality follows from a simple calculation.
Combining the above equations yields the lemma. 
\end{proof}

\begin{lemma}
    \label{lem:ratio-bd2}
    If $1 \leq q^2 < r$, we have $\E [\frac{1}{S'}] \leq  2 q \ln (k+1)$.
\end{lemma}
\begin{proof}
Observe that $\Pr[S' \geq x] = \left(\frac{1 - x}{1 - 1/r}\right)^q$. Thus, the pdf of $S'$ is $\frac{q(1 - x)^{q-1}}{(1 - 1 / r)^q}$ where $x \in [1/r, 1]$. For brevity, we omit the denominator in the following equations and bring it back at the end. 
\begin{align*}
&\E [\frac{1}{S'}] = \int_{x = 1 / r}^1 \frac{1}{x} q(1 - x)^{q-1} \dd x  \\
=& \int_{x = 1 / r}^{1 / q} \frac{1}{x} q(1 - x)^{q-1} \dd x + \int_{x = 1 / q}^{1} \frac{1}{x} q(1 - x)^{q-1} \dd x  \\
\leq& q \int_{x = 1 / r}^{1 / q} \frac{1}{x}  \dd x +   \int_{x = 1 / q}^{1} q\cdot  q(1 - x)^{q-1} \dd x \\
\leq& q \ln r / q  +  q \leq q \ln r \leq q \ln (k+1).
\end{align*}
Thus, by factoring in the denominator $\frac{1}{(1 - 1/r)^q} \leq \frac{1}{(1 - 1/(q^2 +1))^{q}} \leq 2$, we obtain the lemma. 
\end{proof}

\begin{corollary}
    \label{cor:ratio-bd1}
    $\E [\frac{R_i}{R_{i+1}}] \leq 12 q\ln (k+1) \E [Q_{i+1}]$.
\end{corollary}
\begin{proof}
 From Lemmas~\ref{lem:ratio-bd1} and \ref{lem:ratio-bd2}, we know $\E_{\pi([r-1])} [\frac{r}{S(r,q, \pi)}] \leq 12 q\ln (k+1)$. Here, the parameters $r, q, \pi$ are used to make the dependency of $S$ clear. As this holds for any $R_i = r$ and $Q_{i+1} = q$, by deconditioning, we have the desired result. 
\end{proof}

\begin{lemma}
    \label{lem:lb-conditional}
    If $\chain$ has length $M$ and for any integer $c > 0$, we have 
    $\E [\sum_{i = 1}^c Q_i \; | \; M = c] \geq \frac{1}{12 \ln (k+1)} c k^{1/c}.$
\end{lemma}
\begin{proof}
Using the linearly of expectation and Corollary ~\ref{cor:ratio-bd1}, 
\begin{align*}
    &\quad 12 \ln (k+1) \sum_{i = 1}^c \E [Q_i \; | \; t = c] \\
    &\geq \E [\sum_{i = 1}^c \frac{R_{i-1}}{R_{i}} \; | \; R_c = 1, R_{c-1} > 1] \\
    &\geq  \E [c( \prod_{i = 1}^c \frac{R_{i-1}}{R_{i}})^{1/c}
    \; | \; R_c = 1, R_{c-1} > 1] \\
    &=  (k+1)^{1/c},
\end{align*}
where the last inequality follows from the AM--GM inequality and the last equality follows from a
telescoping product, $R_0 = k+1$, and $R_c = 1$. 
\end{proof}

By Lemma~\ref{lem:lb-conditional}, if  algorithm $\cA$ makes at most $c$ cache misses in expectation, then the number of queries it makes is lower bounded by the optimum objective of the following LP:
\begin{align}
    \frac{1}{12 \ln (k+1)} \min &\sum_{i \geq 1} i  k^{1/ i} x_i \nonumber \\
        \sum_{i \geq 1} i x_i &\leq c  \label{con:budget}\\
    \sum_{i \geq 1} x_i &= 1 \nonumber \\    
    x_i &\geq 0 \quad \forall i \geq 1 \nonumber
\end{align}
Here, $x_i := \Pr[M = i]$, i.e., the probability that the chain $\chain$ has length $i$, or equivalently $\cA$ makes $i$ cache misses. The last two constraints define a probability distribution over the values $M$ can have and constraint~(\ref{con:budget}) means that we can afford to make at most $c$ cache misses in expectation. 

\begin{lemma}
    If $c \leq \ln k$, then the above LP's optimum objective is at least $\frac{1}{12 \ln(k+1)} c k^{1/c}$.
\end{lemma}
\begin{proof}
    Let $f(y) := y k^{1/y}$. By simple calculus, we have
    $f'(y) = k^{1/y}(1 - \frac{\ln k}{y})$ and $f''(y) = \frac{\ln^2 k}{y^2} k^{1/ y}$. Thus, $f$ decreases in $y$ for $y \in [1, \ln k]$ and is convex. Then, the LP objective is $\frac{1}{12 \ln (k+1)} \sum_{i \geq 1} f(i) x_i$. By convexity, we have $\sum_{i \geq 1} f(i) x_i  \geq f( \sum_{i \geq 1} i x_i)$. Then, by constraint~(\ref{con:budget}) and $f$ being decreasing in $y$, we have  $f( \sum_{i \geq 1} i x_i) \geq f(c)$. 
\end{proof}

To summarize, we have shown that any $c$-competitive algorithm must make at least $\frac{1}{12 \ln (k+1)}c k^{1/ c}$ cache misses, but under the assumption stated in Lemma~\ref{lem:quadratic}, i.e., $R_i > Q_{i+1}^2$ for all $i$. Thus, by the lemma, we have the following.

\begin{theorem}
    For any integer $c \leq \ln k$, any $c+4$-competitive algorithm must make at least $\frac{1}{12 \ln (k+1)}c k^{1/ c} \cdot \opt$ queries. 
\end{theorem}

\end{document}